 \documentclass[final,5p,times,twocolumn]{elsarticle}

\usepackage{amssymb}
\usepackage{lipsum}
\usepackage{bbm}
\usepackage{amsmath}
\usepackage{braket}
\usepackage{mathtools,slashed}
\usepackage{hyperref}
\usepackage{xcolor}
\usepackage{graphicx,subfigure}
\usepackage{tikz}
\usepackage{physics}
\usepackage{blindtext}
\usepackage{bm}
\usepackage{array}
\usetikzlibrary{calc}
\usepackage[bb=boondox]{mathalfa}
\usepackage{tcolorbox}
\usepackage{rotating}
\usepackage{afterpage}
\usepackage{setspace}
\usepackage{amsthm}
\newcommand{\sM}{\mathcal{M}}
\newcommand{\sH}{\mathcal{H}}
\newcommand{\sA}{\mathcal{A}}

\newcommand{\bC}{\mathbb{C}}

\usepackage{xcolor}

\usepackage{slashed}
\theoremstyle{plain}
\newtheorem{theorem}{Theorem}[section]
\newtheorem{lemma}[theorem]{Lemma}
\newtheorem{proposition}[theorem]{Proposition}
\newtheorem{corollary}[theorem]{Corollary}

\usepackage{bbm}
\theoremstyle{definition}

\journal{Physics Letters B}

\begin{document}

\begin{frontmatter}
\title{Diffeomorphism Invariant Formulation of CP Violation}

\author[first]{Alibordi Muhammad}
\affiliation[first]{organization={Faculty of Physics, University of Warsaw},
            addressline={Pasteura 5}, 
            city={ Warsaw},
            postcode={02-093}, 
            country={Poland}}

\begin{abstract}
We identify a fundamental tension between the standard formulation of CP violation and diffeomorphism invariance in general relativity. The effective Hamiltonian approach, while phenomenologically successful, relies on a preferred time foliation that is incompatible with general covariance. The CP-violating phases are scalars along the worldline of the decaying parent particle; however, the definition of masses and phases presupposes a local covariant structure, which becomes ill-defined near the origin where curvature is large and metric fluctuations become significant. We propose an information-geometric framework based on relative entropy, exploiting pure quantum states in particle and antiparticle Hilbert spaces. We show how the Sakharov conditions could be reinterpreted in terms of information-geometric quantities, although a fully rigorous phenomenological implementation remains to be developed.
\end{abstract}

\end{frontmatter}

\section{SM CPV Arises from Local Field Structures}
\label{introduction}
The deep tension underlying CP violation becomes visible only when cosmology and quantum theory are placed on equal footing. From Friedmann’s discovery that spacetime manifold has a boundary~\cite{Friedma1922}, through the Hawking–Penrose theorems establishing that singular boundaries are  generic~\cite{HawkingPen1970}, to Vilenkin~\cite{VILENKIN198225} and Hartle–Hawking~\cite{PhysRevD282960} where the universe originates from a timeless quantum state without initial data, all modern pictures of the early universe \citep{HollandsWald2015,HawkingPathInt} lack the structures required by the CP violation : Lorentzian time, unique vacuum, causal propagation, particle–antiparticle distinction, and local QFT. Even Sakharov’s conditions~\cite{Sakharov1967} for baryogenesis—departure from equilibrium, baryon-number violation, and C/CP violation—implicitly assume these structures, which do not exist near the origin.  The definition of non-primitive CP violation sourced by the complex phase~\cite{RevModPhys.88.045002} in the CKM(Cabibbo-Kobayashi-Maskawa)~\cite{Cab1963,KobaMas1973} matrix within an essentially-flat spacetime local QFT presupposes a differentiable manifold, causal order, a measure, a symplectic form, a von-Neumann algebra of observables, analyticity, a Hamiltonian generator, a vacuum state, particle-antiparticle structure, a temporal arrow, well-defined scattering phases, background independence, a valid classical limit, unitarity, and decoherence---the absence of any one renders CP violation ill-defined. Fitch-Cronin~\cite{PhysRevLett.13.138} demonstrated CP violation experimentally in neutral kaon decays, and modern measurements such as CMS $\mathrm{B_s \to J/\psi\phi}$~\cite{cpvcms2024} confirm it with high precision.

To give a apparent geometric construct of the Standard Model(SM)~\cite{RevModPhys.46.7,Haag1996,PCT2000,Frankel2012,Hamilton2017,RevModPhys.90.045003,Nakahara2003} on spacetime, begin with a smooth, oriented, time-oriented 4-dimensional Lorentzian manifold $(\mathcal{M},\eta)$ and a maximal SM atlas
$\mathrm{\Phi_{\mathrm{SM}} = \{ (U_\alpha, \Psi_\alpha) \}_{\alpha\in I}}$, where each chart $\mathrm{\Psi_\alpha: U_\alpha \longrightarrow \mathbb{R}^4 \times \mathbb{C}^{d_\mathrm{SM}}}$ maps spacetime coordinates together with internal gauge and spinor degrees of freedom to a local trivialization of the associated Standard Model bundle. The overlaps $\mathrm{U_\alpha\cap U_\beta \neq \varnothing}$ are related by smooth transition functions of class $\mathrm{C^\infty}$, $\mathrm{\Psi_\beta \circ \Psi_\alpha^{-1} :\Psi_\alpha(U_\alpha\cap U_\beta) \longrightarrow \Psi_\beta(U_\alpha\cap U_\beta)}$, compatible with the full local gauge group $\mathrm{\mathsf{G}_{\mathrm{SM}} =  SU(3)_c \times SU(2)_L \times U(1)_Y}$. Let $\mathrm{\mathsf{P}_{\mathrm{SM}} \to \mathcal{M}}$ denote the principal Standard Model bundle with structure group $\mathrm{\mathsf{G}_{\mathrm{SM}}}$, and define the associated vector bundle $\mathrm{\mathcal{E}_{\mathrm{SM}} = \mathsf{P}_{\mathrm{SM}} \times_\rho V_{\mathrm{SM}}}$, where $\mathrm{\rho: \mathsf{G}_{\mathrm{SM}} \to GL(V_{\mathrm{SM}})}$ is the combined representation carrying spinor, flavor, and gauge quantum numbers, and $\mathrm{V_{\mathrm{SM}}}$ is the corresponding finite-dimensional complex vector space. Sections of this bundle $\mathrm{\Psi: \mathcal{M} \longrightarrow \mathcal{E}_{\mathrm{SM}}}$ 
represent all Standard Model matter fields (quarks, leptons, Higgs) simultaneously, and each local chart $\mathrm{\Psi_\alpha}$ provides a trivialization of $\mathrm{\mathcal{E}_{\mathrm{SM}}}$ over $\mathrm{U_\alpha}$. The consistency condition on overlaps ensures that for $\mathrm{x \in U_\alpha \cap U_\beta}$, $\mathrm{\Psi_\beta(x) = \rho(\eta_{\alpha\beta}(x)) \Psi_\alpha(x), \quad \eta_{\alpha\beta}: U_\alpha\cap U_\beta \to \mathsf{G}_{\mathrm{SM}}}$, where $\mathrm{\eta_{\alpha\beta}}$ are the transition functions defining the principal bundle $\mathrm{\mathsf{P}_{\mathrm{SM}}}$. The cocycle condition $\mathrm{\eta_{\alpha\beta} \eta_{\beta\gamma} \eta_{\gamma\alpha} = \mathbf{1}}$ on triple overlaps $\mathrm{U_\alpha \cap U_\beta \cap U_\gamma}$ ensures global consistency. A covariant derivative compatible with the full Standard Model structure is then a connection  $\mathrm{\mathsf{D} : \Gamma(\mathcal{E}_{\mathrm{SM}}) \longrightarrow \Gamma(T^*\mathcal{M}\otimes \mathcal{E}_{\mathrm{SM}})}$, $\mathrm{\mathsf{D}_\mu = \nabla_\mu + i g_s G_\mu^A T^A + i g W_\mu^a \tau^a + i g' B_\mu Y}$, where $\mathrm{\nabla_\mu}$ is the Levi-Civita spin connection, $\mathrm{(G_\mu^A,W_\mu^a,B_\mu)}$ are the gauge connections in their respective Lie algebras, and $\mathrm{(T^A,\tau^a,Y)}$ are the representation matrices acting on $\mathrm{V_{\mathrm{SM}}}$. The sections $\Psi$ transform covariantly under local $\mathsf{G}_{\mathrm{SM}}$ gauge transformations, and the curvature of $\mathsf{D}$ encodes both strong and electroweak field strengths as differential 2-forms $\mathrm{\mathsf{F} = \mathsf{D} \wedge \mathsf{D} \in \Omega^2(\mathcal{M}, \operatorname{End}(\mathcal{E}_{\mathrm{SM}}))}$. The quantum theory on this bundle is perturbatively renormalizable- all ultraviolet divergences arising from loop corrections to Green functions on $(\sM, \eta)$ are absorbed order by order into a finite set of local counterterms $\delta\mathcal{L} \in \Gamma(\Lambda^4 T^*\sM)$($\mathrm{\Lambda}$ being cut-off) constructed covariantly from the metric, curvature, and matter fields, preserving the gauge structure $\mathsf{G}_{\mathrm{SM}}$ and the fiber bundle topology of $\mathcal{E}_{\mathrm{SM}}$ at each energy scale $\mu$. In curved spacetime, this extends through the Hollands--Wald axioms~\cite{HollandsWald2001, HollandsWald2002}: Wick polynomials and time-ordered products are defined locally and covariantly, with renormalization freedom parametrized by local curvature terms that modify the differential operator governing propagation, acting exclusively on the singular (Hadamard) part of the two-point function.

CP violation is defined geometrically as the obstruction to equivariance of dynamics with respect to the involutive~\cite{GRIMUS1997239} bundle automorphism $\mathrm{\Theta_{CP}}$ acting on the Standard Model state bundle $\mathsf{E}\to\mathcal{M}$. Let $\mathsf{F}\subset\mathsf{E}$ be any finite-rank flavor subbundle transported along timelike curves parametrized by proper time $\tau$. Evolution is generated by a non-Hermitian covariant derivative $\mathrm{\mathsf{D}_\tau=\partial_\tau-\mathsf{H},\,  \mathsf{H}= iM-\tfrac12\Gamma}$, with $\mathrm{M}$ and $\Gamma$ Hermitian endomorphisms. CP invariance holds if and only if $\mathrm{[\Theta_{CP},\mathsf{D}_\tau]=0}$,
so failure of this commutation relation provides an intrinsic, coordinate-independent definition of CP violation. Equivalently, CP violation is present whenever the flavor connection $\mathfrak{G}$ on $\mathsf{E}$ has complex holonomy along a closed curve $\gamma$, $\mathrm{\mathcal{U}(\gamma)=\mathcal{P}\exp\!\left(i\oint_\gamma \mathfrak{G}\right)}$,  whose conjugacy class is not invariant under $\mathrm{\Theta_{CP}}$, or when the curvature two-form $\mathsf{F}$ admits a CP-odd invariant, which is equivalent to the connection being non-Abelian in a physically non-removable sense. In any two-state subsystem with transport eigenvalues $\mathrm{\lambda_i=M_i-\tfrac{i}{2}\Gamma_i}$,
mixing induces CP violation whenever $\mathrm{\left|\frac{q}{p}\right|\neq1}$. For transition sections $\mathrm{\mathsf{A}_f}$ the invariant ratio $\mathrm{\lambda_f=\frac{q}{p}\frac{\bar{\mathsf{A}}_f}{\mathsf{A}_f}}$ encodes all observable CP asymmetries, which are scalar functions of $\mathrm{\lambda_f}$. Hence CP violation is equivalently characterized as a non-vanishing CP-odd curvature invariant of the SM flavor bundle over spacetime, implying that it is not fundamentally a property of any specific process or state but a global geometric feature of the gauge–flavor connection, while experimental observables constitute local probes of this underlying curvature obstruction. In local quantum field theory this notion presupposes a smooth, oriented, time-oriented Lorentzian manifold $(\sM,\eta)$ providing causal and differentiable structure, a symplectic form $\Omega$ on classical phase space enabling quantization, a net of von Neumann algebras $\mathcal{O}\mapsto\mathcal{A}(\mathcal{O})$ satisfying isotony and Einstein causality, a unique Poincar\'e-invariant vacuum $\ket{0}$ selected by the spectrum condition and clustering, a self-adjoint Hamiltonian $\mathsf{H}$ generating unitary evolution $\mathrm{U(t)=e^{-i\mathsf{H}t}}$, particle–antiparticle structure ensured by the CPT theorem, well-defined scattering phases encoded in an $\mathrm{S}$-matrix $\mathrm{S_{fi}=\delta_{fi}+i(2\pi)^4\delta^{(4)}(p_f-p_i)\mathcal{M}}$, and environment-induced decoherence selecting a classical limit in which particle trajectories and detector records are operationally well defined.

\section{CP Violation in Curved Spacetime: Structural Failures}
\label{sec:global}

\subsection{Vacuum Uniqueness}
The Poincar\'e-invariant vacuum $\mathrm{\ket{0_P}}$ satisfying $\mathrm{U(\mathfrak{a},\Lambda)\ket{0_P} = \ket{0_P}}$ provides the foundation for SM CP phenomenology, yet the Haag-Kastler~\cite{HaagKastler1964} algebraic framework exposes this uniqueness as contingent rather than fundamental. The net of local algebras $\mathrm{\mathcal{O} \mapsto \sA(\mathcal{O})}$ obeying isotony, Einstein causality, and Poincar\'e covariance does not by itself single out a preferred vacuum. States are positive normalized functionals
$\mathrm{\omega: \sA \to \bC}$ with $\mathrm{\omega(A^*A) \geq 0}$ and $\omega(\mathbbm{1}) = 1$, each inducing via the GNS~\cite{GelfandNaimark1943} construction a representation triple $(\pi_\omega, \sH_\omega, \ket{\Omega_\omega})$ satisfying $\mathrm{\omega(A) = \bra{\Omega_\omega}\pi_\omega(A)\ket{\Omega_\omega}}$. Inequivalent states yield unitarily inequivalent representations $\pi_{\omega_1} \not\cong \pi_{\omega_2}$, making vacuum selection a representation-theoretic choice. Uniqueness is restored only by imposing \emph{global} constraints; the spectrum condition $\mathrm{P^\mu P_\mu \geq 0}$, weak clustering $\mathrm{\lim_{|x|\to\infty}\omega(A\,\alpha_x(B)) = \omega(A)\omega(B)}$, Haag duality~\cite{Haag1996} $\mathrm{\sA(\mathcal{O}')' = \sA(\mathcal{O})}$, and nuclearity. The Borchers-Buchholz~\cite{Borchers1962,Borchers1992,Buchholz1999,PhysRevD.65.085023,Mund2012} theorem then singles out the Poincar\'e vacuum up to unitary equivalence. Vacuum uniqueness thus emerges from global phase-space properties, not from locality alone. The GNS Hilbert space further supports Tomita-Takesaki~\cite{Tomita1967,Takesaki1970} modular theory: the antilinear map $\mathrm{S_{anti}: A\ket{\Omega_\omega} \mapsto A^*\ket{\Omega_\omega}}$ admits polar decomposition $\mathrm{S_{anti} = J\Delta^{1/2}}$ with modular operator $\Delta$ generating the modular flow $\mathrm{\sigma_\tau(A) = \Delta^{i\tau}A\Delta^{-i\tau}}$, which coincides with physical time evolution iff a global timelike Killing vector exists - a structure absent in generic curved spacetimes.

\subsection{Structural Disintegration in Curved Spacetime}
In curved spacetime, every global condition that secures vacuum uniqueness fails simultaneously. Consider de-Sitter spacetime with metric $\mathrm{ds^2 = a^2(\eta)(-d\eta^2 + d\mathbf{x}^2)}$, $\mathrm{a(\eta) = -1/(H\eta)}$ admits no global timelike translation symmetry, no preferred energy operator exists; accordingly the spectrum condition $\mathrm{P^\mu P_\mu \geq 0}$ is undefined, and clustering fails in the presence of cosmological horizons~\cite{PhysRevD.15.2738}. The canonical symplectic form $\mathrm{\Omega = \int d^3x\,\delta\pi^{ij} \wedge \delta h_{ij}}$, conjugating spatial metric $h_{ij}$ to momentum $\mathrm{\pi^{ij} = \sqrt{h}(\mathcal{K}^{ij} - h^{ij}\mathcal{K})}$~\cite{PhysRev.160.1113}($\mathcal{K}$ being curvature tensor), degenerates on the constraint  surface where the Hamiltonian constraint $\mathrm{{}^{(c)}\mathsf{H} = h^{-1/2}(\pi^{ij}\pi_{ij} - \pi^2/2) -\sqrt{h}\,{}^{(3)}R \approx 0}$ and diffeomorphism constraint $\mathrm{{}^{(c)}\mathsf{H}_i = -2\,{}^{(3)}\nabla_j\pi^j{}_i \approx 0}$ hold. No inverse $\mathrm{\Omega^{-1}}$ exists on the constraint surface, and the presymplectic form must be reduced: the physical phase space $\mathrm{\Gamma_{phys} = {}^{(c)}\mathsf{H}^{-1}(0)/\mathbf{Diff}(\Sigma)}$ carries a nondegenerate reduced symplectic form $\mathrm{\Omega_{phys}}$, but different gauge choices yield symplectomorphic yet numerically distinct reduced spaces. Without a unique canonical structure, the discrete symmetries $\mathrm{C,P,T}$ cease to exist as representation-independent operators. The ADM Hamiltonian $\mathrm{\mathsf{H}_{ADM} = \int d^3x\,[N{}^{(c)}\mathsf{H} + N^i{}^{(c)}\mathsf{H}_i]}$~\cite{Arnowitt2008} vanishes weakly on-shell, requiring that temporal evolution be formulated \emph{relationally}: the relational dynamical phase accumulated by the mass eigenstate splitting  becomes $\mathrm{\phi^{\mathrm{rel}} \sim  \Delta m \cdot \mathbf{T}[\mathrm{clock}]}$, where $\mathbf{T}$ is a Dirac observable serving as a physical clock. Different clock choices yield different values of $\mathrm{\phi^{rel}}$, introducing anambiguity absent in the information-geometric formulation developed in Section~\ref{sec:infogeoformulation}. The S-matrix $\mathrm{S_{fi} = \delta_{fi} + i(2\pi)^4\delta^{(4)}(p_f - p_i)\mathcal{M}(i \to f)}$ requires asymptotic states $\mathrm{\ket{i,\mathrm{in}}_{t \to -\infty}}$ and $\mathrm{\ket{f,\mathrm{out}}_{t \to +\infty}}$ where interactions vanish~\cite{BirrellDavies1982,Thiemann2007,ParkerToms2009} - structure absent in expanding universes lacking asymptotic flatness, rendering scattering amplitudes and the CP asymmetries $\mathrm{|\mathcal{M}|^2 - |\bar{\mathcal{M}}|^2 \propto \sin(\Delta\delta)\sin(\Delta\phi)}$ operationally undefined.

\subsection{Hadamard States and State-Dependent Mixing}
Despite the failure of global structures, local quantum field theory remains well-defined through Hadamard states~\cite{BirrellDavies1982,Thiemann2007,ParkerToms2009,Wald1994}. We restrict attention to QFT on fixed globally hyperbolic backgrounds, isolating effects intrinsic to curved spacetime from full quantum gravity. A Hadamard state $\omega$ on $\mathrm{(\sM, g)}$ is characterized by a two-point function whose singular structure matches the Minkowski vacuum up to smooth corrections~\cite{Radzikowski1996,Moretti2021}. For fermions with non-minimal curvature coupling, the Dirac Hadamard parametrix~\cite{Kratzert2000} takes the form
\begin{equation}
\begin{aligned}
\mathrm{S_{Had}(x,x')} &= \mathrm{\frac{i}{8\pi^2}\bigg[\frac{\slashed{\nabla}\sigma(x,x')}{\sigma(x,x')^2}\,U(x,x')}\\
&\mathrm{\quad + V(x,x')\ln\!\left(\frac{\sigma(x,x')}{\ell^2}\right)+ W(x,x')\bigg]}
\end{aligned}
\label{eq:hadamard_parametrix}
\end{equation}
where, $\mathrm{\sigma(x,x')}$ is Synge's world function, $\mathrm{U(x,x')=\Delta^{1/2}(x,x')\,\mathcal G(x,x')}$ with $\mathrm{\mathcal{G}(x,x')}$ being the bispinor parallel propagator along the geodesic connecting $\mathrm{x}$ and  $\mathrm{\Delta^{1/2}}$ is the Van Vleck-Morette determinant, $\mathrm{V(x,x')}$ is determined by the DeWitt recursion relations~\cite{FULLING199073}, and $\ell$ is a renormalization length scale. The singular terms ($\mathrm{\sigma^{-2}}$, logarithmic) are geometrically determined and state-independent, while $\mathrm{W(x,x')}$ encodes the vacuum choice.  The class of physically admissible Hadamard states comprises the Bunch--Davies state $\mathrm{W_{BD}}$~\cite{BunchDavies1978}, which is the unique de Sitter--invariant state; the adiabatic vacua $W^{(n)}$ with $n\ge2$, obtained via a WKB expansion and satisfying the Hadamard condition by Junker's theorem~\cite{Junker1996}; the Hartle--Hawking state $\mathrm{W_{HH}}$~\cite{PhysRevD.13.2188}, a KMS(Kubo–Martin–Schwinger) state at the Hawking temperature $\mathrm{T_H=\alpha/(2\pi)}$; and the Unruh state $\mathrm{W_{U}}$~\cite{PhysRevD.14.870}, which is thermal on the future horizon. These states define unitarily inequivalent GNS representations~\cite{Haag1996}. Critically, the mixing matrix element
\begin{equation}
\mathrm{M_{12}[\omega] \sim \omega\!\left(\mathcal{T}\left\{\mathcal{H}_{\mathrm{eff}}^{\Delta B=2}(x)\,\mathcal{H}_{\mathrm{eff}}^{\Delta B=2}(y)\right\}\right)}
\label{eq:M12_state_dep}
\end{equation}
depends ($\mathcal{H}_{\mathrm{eff}}$ being weak effective Hamiltonian) on the smooth part $\mathrm{W(x,x')}$ of the two-point function, making $\mathrm{q/p}$ and hence $\mathrm{\phi = \arg(-M_{12}/\Gamma_{12})}$ state-dependent. The state-dependent part is a smooth bisolution $\mathrm{\Delta W(x,y) = W_{\omega_1}(x,y) - W_{\omega_2}(x,y)}$ satisfying $\mathrm{P_x\,\Delta W(x,y) = 0}$ in both arguments, where $\mathrm{P = -\Box_g + m^2 + \xi R}$ is the Klein--Gordon operator. This cannot be absorbed into local counterterms; any counterterm modifies the operator $\mathrm{P \to P' = P + \delta P}$, inducing a two-point function shift $\mathrm{\delta_{\mathrm{ct}}W}$ that satisfies the \emph{inhomogeneous} equation
$\mathrm{P_x\,\delta_{\mathrm{ct}}W = -\delta P_x\, G_P \neq 0}$. Since $\mathrm{\Delta W \in \ker(P)}$ while $\mathrm{\delta_{\mathrm{ct}}W \notin \ker(P)}$, no choice of renormalization scheme can eliminate the state-dependent
contribution. The shift $\mathrm{\Delta M_{12}}$ is therefore a genuine physical effect surviving renormalization~\cite{HollandsWald2001}. Flavor dependence $\mathrm{W_b(m_b) \neq W_s(m_s)}$ induces mass-dependent vacuum
particle creation rates that modify the \emph{magnitude} of CP-violating observables through state-dependent loop corrections; however, CPV itself continues to require the CKM phase~\cite{KobaMas1973} -  curvature alone does not generate CP asymmetry. The central conclusion is that curved spacetime destroys not merely vacuum uniqueness but the kinematical structures on which SM CPV violation is built , (a) symplectic non-degeneracy enabling canonical quantization, (b) a Hamiltonian generator producing unambiguous time flow, (c) Wick rotation defining particles, (d) asymptotic states constructing the S-matrix, and,  (e) foliation-independent phases. Standard Model CP violation is therefore an emergent semiclassical phenomenon, contingent on Minkowski's global symmetries and inapplicable in its standard form to quantum gravitational or early-universe epochs. What survives is the local algebraic structure and the Hadamard condition - and it is from these alone that a diffeomorphism invariant formulation must be constructed.

\section{Diffeomorphism Invariant Reformulation}
\label{sec:infogeoformulation}

\subsection{Information-Geometric Formulation of CP Asymmetry}
The standard formulation of CP violation presupposes a local spacetime structure whose global extension renders the observable ill-defined. We resolve this by proposing an information-geometric~\cite{AmariNagaoka2000} reformulation in which the asymmetry is defined intrinsically through the quantum relative entropy between particle and antiparticle density states. Let $\rho$ and $\bar{\rho}$ denote the reduced density matrices~\cite{NielsenChuang2010} of the particle and antiparticle respectively, obtained by tracing over unobserved degrees of freedom from the full density operator, 
\begin{equation}
\rho = \mathrm{Tr}_{\mathrm{unobs}}\!\left[\rho^{\mathrm{full}}\right], \qquad
\bar{\rho} = \mathrm{Tr}_{\mathrm{unobs}}\!\left[\bar{\rho}^{\mathrm{full}}\right].
\end{equation}
This coarse-graining renders both states generically mixed, satisfying $\mathrm{Tr}(\rho^2) < 1$, which necessitates the full apparatus of mixed-state information geometry. The modular operators $K = -\log\rho$ and $\bar{K} = -\log\bar{\rho}$ are then well-defined, and the quantum relative entropies take the form, 
\begin{equation}
S(\rho\,||\,\bar{\rho}) = \mathrm{Tr}(\rho\bar{K}) - \mathrm{Tr}(\rho K), \quad S(\bar{\rho}\,||\,\rho) = \mathrm{Tr}(\bar{\rho}K) - \mathrm{Tr}(\bar{\rho}\bar{K})
\label{eq:relen}
\end{equation}
where, $\mathrm{Tr}(\rho K)$ is the self-informative von Neumann entropy and $\mathrm{Tr}(\rho\bar{K})$ is the cross-entropy in which information is perceived through the wrong model~\cite{Umegaki1954,Araki1975}. The information-geometric CP asymmetry is
then defined as,
\begin{equation}
\mathfrak{A}_{\mathrm{info}}[g] = \frac{S(\rho\,||\,\bar{\rho})- S(\bar{\rho}\,||\,\rho)}{S(\rho\,||\,\bar{\rho}) + S(\bar{\rho}\,||\,\rho)}.
\label{eq:Ainfo_definition}
\end{equation}
The geometric content is naturally organized by the pair of modular operators $\{K, \bar{K}\}$~\cite{Haag1996}, which simultaneously encode two essential ingredients: spectral asymmetry $\mathrm{Spec}(K) \neq \mathrm{Spec}(\bar{K})$, quantifying statistical
distinguishability between the states, and non-commutativity $[K,\bar{K}] \neq 0$, which constitutes the microscopic origin of Berry curvature. The antisymmetric relative entropy thereby provides a direct operational bridge between modular
non-commutativity and geometric phase, elevating Berry like curvature~\cite{Berry1984} from a kinematic notion to an information-theoretic observable. The CP violation structure requires comparing two distinct parametric families, $\rho(\bm{\theta}): \bm{\theta} \in \mathcal{M} \subset \mathbb{R}^n$ and $\bar{\rho}(\bm{\theta}'): \bm{\theta}' \in \bar{\mathcal{M}} \subset \mathbb{R}^n$. The full geometric structure lives on the product manifold $\mathcal{M}_{\mathrm{CP}} = \mathcal{M} \times \bar{\mathcal{M}}$, and the physical CP violation is assessed by restricting to the diagonal embedding $\Delta: \mathcal{M} \to \mathcal{M}_{\mathrm{CP}}$, $\bm{\theta} \mapsto (\bm{\theta}, \bm{\theta})$, where $\rho(\bm{\theta})$
and $\bar{\rho}(\bm{\theta})$ are compared at identical parameter values. Expanding $S(\rho(\bm{\theta})||\bar{\rho}(\bm{\theta}'))$ around $\bm{\theta}' = \bm{\theta}$,
\begin{equation}
\begin{aligned}
S(\rho(\bm{\theta})||\bar{\rho}(\bm{\theta} + \delta\bm{\theta}))
&= S(\rho\,||\,\bar{\rho})\Big|_{\bm{\theta}} + \left[\frac{\partial S(\rho\,||\,\bar{\rho})}{\partial\theta'^i} \right]_{\theta'=\theta}\!\delta\theta^i\\
&\quad + \frac{1}{2}\left[\frac{\partial^2 S(\rho\,||\,\bar{\rho})}{\partial\theta'^i\partial\theta'^j}\right]_{\theta'=\theta} \!\delta\theta^i\delta\theta^j + \mathcal{O}(\delta\theta^3).
\end{aligned}
\label{eq:cross_expansion}
\end{equation}
The zeroth-order term $S(\rho(\bm{\theta})||\bar{\rho}(\bm{\theta})) > 0$ is the base-point distinguishability responsible for CP violation. The first-order term $\nabla_i^{(\bar{\rho})}S(\rho\,||\,\bar{\rho}) = \mathrm{Tr}[\rho\,\partial_i\bar{K}]$ is the gradient of relative entropy. The second-order information response kernel, 
\begin{equation}
\begin{aligned}
\mathbf{D}_{\mathrm{cross}}(\bm{\theta}) &= \frac{\partial^2}{\partial\theta'^i\partial\theta'^j} \left[S(\rho(\bm{\theta})||\bar{\rho}(\bm{\theta}'))\right]_{\theta'=\theta}\\
&= -\frac{\partial^2}{\partial\theta'^i\partial\theta'^j} \mathrm{Tr}[\rho(\theta)\bar{K}(\theta')]\Big|_{\theta'=\theta}
\end{aligned}
\label{eq:cross_metric_definition}
\end{equation}
does not belong to the standard Fisher metric family since it couples two distinct state manifolds~\cite{PhysRevLett.72.3439,Nielsen2020}. It reduces to the Fisher metric in the CP-conserving limit, $\lim_{\phi\to 0}\mathbf{D}_{\mathrm{cross}}(\bm{\theta})
= \mathbf{D}_{\mathrm{cross}}^{\mathrm{Fisher}}[\rho](\bm{\theta})$, since $\lim_{\phi\to 0}\bar{\rho}(\bm{\theta}) = \rho(\bm{\theta})$. For finite CP-violating phase $\phi$, using $\bar{\rho}(\phi) = e^{-2i\phi\mathfrak{C}_{\mathrm{CP}}}\rho(0)
e^{2i\phi\mathfrak{C}_{\mathrm{CP}}}$ and the identity for modular operator derivatives, 
\begin{equation}
\frac{\partial^2 K}{\partial\theta^i\partial\theta^j} = -\int_0^\infty ds\,\left[\rho(\theta)^{is}\frac{\partial^2\rho(\theta)}{\partial\theta^i\partial\theta^j}\rho(\theta)^{-is} + (\text{mod. flows})\right]
\end{equation}
one obtains the CP-violating correction to the cross-metric,
\begin{equation}
\Delta\mathbf{D}_{\mathrm{CP}}(\phi) = c\phi\,\mathrm{Re}\,\mathrm{Tr}\left[\rho\int_0^\infty ds\,\rho^{is}[\mathfrak{C}_{\mathrm{CP}},\partial_i\rho]\rho^{-is}\partial_j\rho\right] + \mathcal{O}(\phi^2)
\label{eq:measure}
\end{equation}
where, $\mathfrak{C}_{\mathrm{CP}}$ is the CP generator and $\mathbf{D}_{\mathrm{cross}} = \mathbf{D}^{\mathrm{Fisher}} + \Delta\mathbf{D}_{\mathrm{CP}}(\phi)$.

\subsection{Perturbative Structure of the Symmetric Denominator}
The denominator of $\mathfrak{A}_{\mathrm{info}}$ is the symmetric relative entropy $\mathbf{D}_{\mathrm{sym}} = S(\rho\,||\,\bar{\rho}) + S(\bar{\rho}\,||\,\rho)$, which expands as, 
\begin{equation}
\mathbf{D}_{\mathrm{sym}} = \mathrm{Tr}(\rho\bar{K}) - \mathrm{Tr}(\rho K) + \mathrm{Tr}(\bar{\rho}K) - \mathrm{Tr}(\bar{\rho}\bar{K}).
\end{equation}
To evaluate this systematically, we decompose the density matrices into diagonal and off-diagonal parts, $\rho = \rho^{\mathrm{diag}} + \rho^{\mathrm{off}}$ and $\bar{\rho} = \bar{\rho}^{\mathrm{diag}} + \bar{\rho}^{\mathrm{off}}$, and apply matrix logarithm perturbation theory,
\begin{equation}
\log(\rho^{\mathrm{diag}} + \rho^{\mathrm{off}}) = \log(\rho^{\mathrm{diag}})- \sum_{n=1}^{\infty}\frac{1}{n}\left[(\rho^{\mathrm{diag}})^{-1}\rho^{\mathrm{off}}\right]^n.
\end{equation}
This expansion is controlled by the off-diagonal coherence parameter $\epsilon \equiv \|(\rho^{\mathrm{diag}})^{-1}\rho^{\mathrm{off}}\| < 1$, which quantifies the relative size of interference terms induced by tracing over unobserved degrees of freedom. Introducing the shorthand $X \equiv -(\rho^{\mathrm{diag}})^{-1}\rho^{\mathrm{off}}$ and $\bar{X} \equiv -(\bar{\rho}^{\mathrm{diag}})^{-1}\bar{\rho}^{\mathrm{off}}$,the full perturbative expansion of $\mathbf{D}_{\mathrm{sym}}$ reads in Eq.~\ref{eq:Dsym_full}. The twelve first-order terms organize into three structural types. The Type A terms (1a, 2a, 3a, 4a) involve only diagonal-diagonal logarithmic traces and cancel pairwise since $\rho^{\mathrm{diag}} = \bar{\rho}^{\mathrm{diag}}$ at the diagonal embedding. The Type B terms (1b, 2b, 3b, 4b) of the form $\mathrm{Tr}[\rho^{\mathrm{diag}}(\cdot)^{-1}(\cdot)^{\mathrm{off}}]$ vanish individually since diagonal times off-diagonal matrices have zero trace. The Type C terms (1c, 2c, 3c, 4c) of the form $\mathrm{Tr}[\rho^{\mathrm{off}}\log(\cdot)^{\mathrm{diag}}]$ also vanish for the same reason. Consequently,  leading non-vanishing contribution arises at second order $\mathbf{D}_{\mathrm{sym}} = \mathcal{O}(\epsilon^2)$ from cross-terms of the form $\mathrm{Tr}[\rho^{\mathrm{off}}\bar{X}]$, which couple off-diagonal elements of one state to those of the other and are therefore the lowest-order contributions sensitive to CP-violating phases. The multi-coherence terms at order $\mathcal{O}(\epsilon^n)$ for $n \geq 2$ describe higher-order interference effects and are systematically suppressed in the small-coherence regime. This perturbative hierarchy ensures that $\mathbf{D}_{\mathrm{sym}}$ is positive, well-controlled, and carries a clean physical interpretation. The denominator measures the total quantum-coherence-driven distinguishability between the two state manifolds, against which the antisymmetric CP-violating signal in the numerator is normalized.

\begin{equation}
\begin{aligned}
&\mathbf{D}_{\mathrm{sym}}=\\
& \underbrace{-\mathrm{Tr}[\rho^{\mathrm{diag}}\log(\bar{\rho}^{\mathrm{diag}})]}_{\text{1a}}+ \underbrace{\mathrm{Tr}[\rho^{\mathrm{diag}}(\bar{\rho}^{\mathrm{diag}})^{-1}\bar{\rho}^{\mathrm{off}}]}_{\text{1b}}
- \underbrace{\mathrm{Tr}[\rho^{\mathrm{off}}\log(\bar{\rho}^{\mathrm{diag}})]}_{\text{1c}}\\
& + \underbrace{\mathrm{Tr}[\rho^{\mathrm{diag}}\log(\rho^{\mathrm{diag}})]}_{\text{2a}}- \underbrace{\mathrm{Tr}[\rho^{\mathrm{diag}}(\rho^{\mathrm{diag}})^{-1}\rho^{\mathrm{off}}]}_{\text{2b}}
+ \underbrace{\mathrm{Tr}[\rho^{\mathrm{off}}\log(\rho^{\mathrm{diag}})]}_{\text{2c}}\\
& - \underbrace{\mathrm{Tr}[\bar{\rho}^{\mathrm{diag}}\log(\rho^{\mathrm{diag}})]
}_{\text{3a}}
+ \underbrace{\mathrm{Tr}[\bar{\rho}^{\mathrm{diag}}
(\rho^{\mathrm{diag}})^{-1}\rho^{\mathrm{off}}]}_{\text{3b}}
- \underbrace{\mathrm{Tr}[\bar{\rho}^{\mathrm{off}}\log(\rho^{\mathrm{diag}})]
}_{\text{3c}}\\
& + \underbrace{\mathrm{Tr}[\bar{\rho}^{\mathrm{diag}}
\log(\bar{\rho}^{\mathrm{diag}})]}_{\text{4a}}
- \underbrace{\mathrm{Tr}[\bar{\rho}^{\mathrm{diag}}
(\bar{\rho}^{\mathrm{diag}})^{-1}\bar{\rho}^{\mathrm{off}}]}_{\text{4b}}
+ \underbrace{\mathrm{Tr}[\bar{\rho}^{\mathrm{off}}
\log(\bar{\rho}^{\mathrm{diag}})]}_{\text{4c}}\\
& + \sum_{n=2}^{\infty}\frac{1}{n}\Bigl\{
\underbrace{\mathrm{Tr}[\rho^{\mathrm{diag}}\bar{X}^n]
- \mathrm{Tr}[\rho^{\mathrm{diag}}X^n]}_{\text{diag}\times\text{multi-coh.}}
+ \underbrace{\mathrm{Tr}[\bar{\rho}^{\mathrm{diag}}X^n]
- \mathrm{Tr}[\bar{\rho}^{\mathrm{diag}}\bar{X}^n]}_{\bar{\text{diag}}
\times\text{multi-coh.}}\\
&+ \underbrace{\mathrm{Tr}[\rho^{\mathrm{off}}\bar{X}^n]
- \mathrm{Tr}[\rho^{\mathrm{off}}X^n]}_{\text{off}\times\text{multi-coh.}}
+ \underbrace{\mathrm{Tr}[\bar{\rho}^{\mathrm{off}}X^n]
- \mathrm{Tr}[\bar{\rho}^{\mathrm{off}}\bar{X}^n]}_{\bar{\text{off}}
\times\text{multi-coh.}}\Bigr\}
\end{aligned}
\label{eq:Dsym_full}
\end{equation}

\subsection{Antisymmetric Relative Entropy as Uhlmann Geometric Phase}

CP violation requires comparing states from two distinct Hilbert spaces. The particle state $\rho(\bm{\theta})$ lives on $\mathcal{H}$ and the antiparticle state $\bar{\rho}(\bm{\theta}')$ lives on $\bar{\mathcal{H}}$, with the physical comparison performed on the product space $\mathcal{H}_{\mathrm{CP}} = \mathcal{H} \otimes\bar{\mathcal{H}}$. The corresponding state manifolds $\mathcal{M}$ and $\bar{\mathcal{M}}$each carry their own Uhlmann geometry, with the physical CP structure encoded in the diagonal embedding, 
\begin{equation}
\Delta: \mathcal{M} \to \mathcal{M} \times \bar{\mathcal{M}}, \quad
\bm{\theta} \mapsto (\bm{\theta}, \bm{\theta})
\label{eq:diagonal_embedding}
\end{equation}
where, $\rho(\bm{\theta})$ and $\bar{\rho}(\bm{\theta})$ are compared at identical parameter values. For the particle state with spectral decomposition $\rho = \sum_n p_n\ket{n}\bra{n}$, the purification $\ket{\Psi} = \sum_n \sqrt{p_n}\,\ket{n}_{\mathcal{H}} \otimes \ket{n}_{\mathrm{aux}}$ defines the Uhlmann connection 1-form~\cite{UHLMANN1986229} via parallel transport $\mathrm{\langle \Psi | \delta \Psi \rangle = 0}$, 
\begin{equation}
\Upsilon_i = \mathrm{i}\sum_{m,n}
\frac{\bra{m}\partial_i\rho\ket{n}}{p_m + p_n}\ket{m}\bra{n}
\label{eq:uhlmann_particle}
\end{equation}
with curvature 2-form $\Omega_{ij} = \partial_i\Upsilon_j - \partial_j\Upsilon_i + \mathrm{i}[\Upsilon_i,\Upsilon_j]$~\cite{UHLMANN1993253}. For the antiparticle state $\bar{\rho}(\bm{\theta}')$ with spectral decomposition $\bar{\rho} = \sum_n \bar{p}_n\ket{\bar{n}}\bra{\bar{n}}$, the analogous connection on $\bar{\mathcal{M}}$ is,
\begin{equation}
\bar{\Upsilon}_j = \mathrm{i}\sum_{m,n}
\frac{\bra{\bar{m}}\partial_j'\bar{\rho}\ket{\bar{n}}}{\bar{p}_m + \bar{p}_n}
\ket{\bar{m}}\bra{\bar{n}}
\label{eq:uhlmann_antiparticle}
\end{equation}
where, $\partial_j' = \partial/\partial\theta'^j$, with curvature $\bar{\Omega}_{jk} = \partial_j'\bar{\Upsilon}_k - \partial_k'\bar{\Upsilon}_j+ \mathrm{i}[\bar{\Upsilon}_j,\bar{\Upsilon}_k]$. The antisymmetric relative entropy, which constitutes the numerator of $\mathfrak{A}_{\mathrm{info}}$, is $\mathbf{D}_{\mathrm{asym}} =S(\rho||\bar{\rho}) - S(\bar{\rho}||\rho)$. We establish its geometric character through the following proposition.

\begin{proposition}
For finite displacement from $\bm{\theta}$ to $\bm{\theta}'$ along a path
$\mathcal{C}$ in the product manifold $\mathcal{M} \times \bar{\mathcal{M}}$:
\begin{equation}
\mathbf{D}_{\mathrm{asym}} = 2\,\mathrm{Im}
\int_{\mathcal{C}(\bm{\theta}\to\bm{\theta}')}
\mathrm{Tr}\left[\rho\,\Upsilon_i - \bar{\rho}\,\bar{\Upsilon}_i\right]d\theta^i
+ \mathcal{O}(|\bm{\theta}'-\bm{\theta}|^3)
\label{eq:numerator_geometric}
\end{equation}
\end{proposition}

\begin{proof}
Expanding $S(\rho(\bm{\theta})||\bar{\rho}(\bm{\theta}'))$ around
$\bm{\theta}' = \bm{\theta}$ and using
$\partial_i K_{\bar{\rho}} = -\int_0^\infty ds\,\bar{\rho}^{is}
(\partial_i\bar{\rho})\bar{\rho}^{-is}$, the antisymmetric combination gives,
\begin{equation}
\begin{aligned}
\mathbf{D}_{\mathrm{asym}} &= \mathrm{Tr}[\rho(K_{\bar{\rho}} - K_\rho)]- \mathrm{Tr}[\bar{\rho}(K_\rho - K_{\bar{\rho}})]\\
&= \mathrm{Tr}[(\rho + \bar{\rho})(K_{\bar{\rho}} - K_\rho)].
\end{aligned}
\end{equation}
Applying the identity $\mathrm{Tr}[\rho\log\sigma] - \mathrm{Tr}[\sigma\log\rho] = 2\,\mathrm{Im}\,\mathrm{Tr}[\rho^{1/2}(\log\rho - \log\sigma)\sigma^{1/2}]$ and the relation between the Uhlmann 1-form and the logarithm difference, 
\begin{equation}
\mathrm{Im}\,\mathrm{Tr}\left[\rho(\partial_i\log\rho
- \partial_i\log\bar{\rho})\right]
= \mathrm{Tr}[\rho\,\Upsilon_i - \bar{\rho}\,\bar{\Upsilon}_i]
\label{eq:connection_log_relation}
\end{equation}
integrating along $\mathcal{C}$ yields Eq.~\eqref{eq:numerator_geometric}.
\end{proof}

Three structural properties follow. First, for infinitesimal displacement $d\bm{\theta} \to 0$ we can have, 
\begin{equation}
S(\rho||\rho + d\rho) - S(\rho + d\rho||\rho) = \mathcal{O}(d\theta^3)
\label{eq:infinitesimal_vanish}
\end{equation}
since $\rho^{1/2}d\rho\,\rho^{-1/2}$ is Hermitian and its trace is therefore real. The antisymmetric part thus requires finite loops and vanishes identically at infinitesimal order. Second, on the diagonal $\bm{\theta}' = \bm{\theta}$, the numerator encodes the connection difference $\Delta\Upsilon_i = \Upsilon_i - \bar{\Upsilon}_i$, which measures the eigenbasis mismatch between $\mathcal{M}$ and $\bar{\mathcal{M}}$ generated by the CP-violating phase. Third, this mismatch is non-zero iff the modular operators fail to commute,
\begin{equation}
\Delta\Upsilon_i \neq 0 \Longleftrightarrow  [K, \bar{K}] \neq 0
\label{eq:commutator_connection}
\end{equation}
since $K$ and $\bar{K}$ share an eigenbasis iff $[K,\bar{K}] = 0$, which holds iff  $\phi = 0$. The algebraic mechanism is made transparent by the modular automorphism on $\mathcal{M}_3 = \mathcal{B}(\mathbb{C}^3)$,
\begin{equation}
\sigma_s(O) = e^{isK}Oe^{-isK} = \sum_{n=0}^\infty\frac{(is)^n}{n!} \underbrace{[K,[K,\cdots[K,O]\cdots]]}_{n\,\mathrm{times}}
\label{eq:modular_flow_formal}
\end{equation}
since $[K,\rho] = 0$ and $[\bar{K},\bar{\rho}] = 0$ individually, the asymmetry arises entirely from cross-terms,
\begin{equation}
\mathbf{D}_{\mathrm{asym}} \sim \mathrm{Tr}[\rho\bar{K}] - \mathrm{Tr}[\bar{\rho}K] \neq 0 \Longleftrightarrow  [K, \bar{K}] \neq 0.
\label{eq:cross_term_asymmetry}
\end{equation}
The Frobenius norm $\|[K,\bar{K}]\|_F$ provides a basis-independent, unitarily invariant scalar quantifying this mismatch. The off-diagonal modular gaps $\Delta\kappa_{mn} = \kappa_m - \kappa_n$ govern the oscillatory phases $e^{is\Delta\kappa_{mn}}$ in the modular flow, and the non-commutativity $[K,\bar{K}] \neq 0$ encodes a Berry-like geometric curvature in the helicity state space arising from eigenbasis rotation between the two manifolds. In summary,
\begin{equation}
\mathbf{D}_{\mathrm{asym}} = S(\rho||\bar{\rho}) - S(\bar{\rho}||\rho)
= 2\,\mathrm{Im}\int_{\mathcal{C}_{\phi}}
\mathrm{Tr}[\rho\,\Upsilon_i - \bar{\rho}\,\bar{\Upsilon}_i]\,d\theta^i
\label{eq:numerator_final}
\end{equation}
with geometric origin in the Uhlmann connection mismatch $\Delta\Upsilon = \Upsilon - \bar{\Upsilon}$, algebraic origin in modular non-commutativity $[K,\bar{K}] \neq 0$, and the sharp criterion $\mathbf{D}_{\mathrm{asym}} = 0 \Leftrightarrow \phi = 0$.

\subsection{Diffeomorphism Invariance}
\begin{theorem}[Diffeomorphism invariance of information-geometric CP asymmetry]
Let $(\mathcal{M}, g)$ and $(\mathcal{M}', g')$ be globally hyperbolic spacetimes, and let $\mathfrak{d}: \mathcal{M} \to \mathcal{M}'$ be an isometric diffeomorphism with $\mathfrak{d}^*g' = g$. Let $\alpha_{\mathfrak{d}}: \mathcal{A}[\mathcal{M}]\to \mathcal{A}[\mathcal{M}']$ denote the induced algebra isomorphism in the Brunetti--Fredenhagen--Verch framework~\cite{Brunetti2003}. Then,
\begin{equation}
A_{\mathrm{info}}[\mathfrak{d}^*g'] = A_{\mathrm{info}}[g]
\end{equation}
\end{theorem}

\begin{proof}
By the BFV locally covariant axioms, the quantum state assignment satisfies:
\begin{equation}
\omega'[\mathfrak{d}^*g'] = \omega[g] \circ \alpha_{\mathfrak{d}}^{-1}
\label{eq:bfv_covariance}
\end{equation}
This induces the density matrix transformation:
\begin{equation}
\rho[\mathfrak{d}^*g'] = U_{\mathfrak{d}}\,\rho[g]\,U_{\mathfrak{d}}^\dagger
\label{eq:density_transform}
\end{equation}
where $U_{\mathfrak{d}}$ is the unitary implementer of $\alpha_{\mathfrak{d}}$
in the GNS representation. We first establish the following lemma.

\begin{lemma}
Relative entropy is unitarily invariant. For any unitary $U$:
\begin{equation}
S(U\rho\, U^\dagger\, ||\, U\bar{\rho}\, U^\dagger) = S(\rho\,||\,\bar{\rho})
\label{eq:unitary_invariance_relent}
\end{equation}
\end{lemma}

\begin{proof}
Using $\log(UAU^\dagger) = U(\log A)U^\dagger$:
\begin{equation}
\begin{aligned}
S(U\rho\, U^\dagger\, ||\, U\bar{\rho}\, U^\dagger)
&= \mathrm{Tr}\bigl[(U\rho U^\dagger)
(\log(U\rho U^\dagger) - \log(U\bar{\rho} U^\dagger))\bigr]\\
&= \mathrm{Tr}\bigl[U\rho U^\dagger
(U\log\rho\, U^\dagger - U\log\bar{\rho}\, U^\dagger)\bigr]\\
&= \mathrm{Tr}\bigl[U\rho(\log\rho - \log\bar{\rho})U^\dagger\bigr]\\
&= \mathrm{Tr}\bigl[\rho(\log\rho - \log\bar{\rho})\bigr]
= S(\rho\,||\,\bar{\rho})
\end{aligned}
\end{equation}
where the last step uses cyclicity of the trace.
\end{proof}

Applying Eq.~\eqref{eq:density_transform} and the lemma to both
particle and antiparticle states:
\begin{equation}
\begin{aligned}
S(\rho[\mathfrak{d}^*g']\,||\,\bar{\rho}[\mathfrak{d}^*g'])
&= S(U_{\mathfrak{d}}\,\rho[g]\,U_{\mathfrak{d}}^\dagger\,||\,
U_{\mathfrak{d}}\,\bar{\rho}[g]\,U_{\mathfrak{d}}^\dagger)\\
&= S(\rho[g]\,||\,\bar{\rho}[g])
\end{aligned}
\end{equation}
The same invariance holds for the reversed term
$S(\bar{\rho}[\mathfrak{d}^*g']\,||\,\rho[\mathfrak{d}^*g'])
= S(\bar{\rho}[g]\,||\,\rho[g])$,
and therefore both the numerator $\mathbf{D}_{\mathrm{asym}}$ and the
denominator $\mathbf{D}_{\mathrm{sym}}$ are individually invariant under
$\mathfrak{d}$. Consequently,
\begin{equation}
A_{\mathrm{info}}[\mathfrak{d}^*g']
= \frac{\mathbf{D}_{\mathrm{asym}}[\mathfrak{d}^*g']}
{\mathbf{D}_{\mathrm{sym}}[\mathfrak{d}^*g']}
= \frac{\mathbf{D}_{\mathrm{asym}}[g]}{\mathbf{D}_{\mathrm{sym}}[g]}
= A_{\mathrm{info}}[g]
\end{equation}
\end{proof}

\begin{corollary}[Foliation independence]
Unlike phase-based observables $\phi_s(\Sigma_t)$ that depend on the choice of spacelike hypersurface $\Sigma_t$, the information asymmetry $A_{\mathrm{info}}$ is independent of foliation,
\begin{equation}
A_{\mathrm{info}}[\Sigma_{t_1}] = A_{\mathrm{info}}[\Sigma_{t_2}]
\end{equation}
for any two Cauchy surfaces $\Sigma_{t_1}, \Sigma_{t_2}$ of a globally hyperbolic spacetime $(\mathcal{M}, g)$.
\end{corollary}

\begin{proof}
In a globally hyperbolic spacetime $(\mathcal{M}, g)$, any two Cauchy surfaces $\Sigma_{t_1}$ and $\Sigma_{t_2}$ are related by a diffeomorphism $\mathfrak{d}: \mathcal{M} \to \mathcal{M}$ that is an element of the spacetime diffeomorphism group and satisfies $\mathfrak{d}^*g = g$. By the BFV local covariance axiom, the algebra isomorphism $\alpha_{\mathfrak{d}}: \mathcal{A}[\mathcal{M}] \to \mathcal{A}[\mathcal{M}]$ acts unitarily on the GNS Hilbert space, and the density matrices restricted to $\Sigma_{t_1}$ and $\Sigma_{t_2}$ are therefore related by,
\begin{equation}
\rho[\Sigma_{t_2}] = U_{\mathfrak{d}}\,\rho[\Sigma_{t_1}]\,U_{\mathfrak{d}}^\dagger,
\qquad
\bar{\rho}[\Sigma_{t_2}] = U_{\mathfrak{d}}\,\bar{\rho}[\Sigma_{t_1}]\,U_{\mathfrak{d}}^\dagger
\end{equation}
By Lemma~\ref{eq:unitary_invariance_relent}, unitary invariance of relative entropy then gives,
\begin{equation}
S(\rho[\Sigma_{t_2}]\,||\,\bar{\rho}[\Sigma_{t_2}])= S(\rho[\Sigma_{t_1}]\,||\,\bar{\rho}[\Sigma_{t_1}])
\end{equation}
and identically for the reversed term, so both $\mathbf{D}_{\mathrm{asym}}$ and $\mathbf{D}_{\mathrm{sym}}$ are foliation-independent, and the result $A_{\mathrm{info}}[\Sigma_{t_1}] = A_{\mathrm{info}}[\Sigma_{t_2}]$ follows immediately from the diffeomorphism invariance theorem.
\end{proof}

\subsection{CP Conservation and Modular Thermal Equilibrium}

The diffeomorphism invariance of $\mathfrak{A}_{\mathrm{info}}$ established above admits a deeper thermodynamic interpretation through the modular structure of the state manifolds~\cite{PhysRevLett.117.041601,RevModPhys.74.197}. For a general quantum state $\rho$ with modular Hamiltonian $K = -\log\rho$, the effective modular temperature is defined as, 
\begin{equation}
T_{\mathrm{mod}}[\rho] = \left(\frac{\partial S[\rho]}{\partial\langle H\rangle}
\right)^{-1}
\label{eq:modular_temp_general}
\end{equation}
which, for a state diagonal in its eigenbasis $\rho \approx \sum_n p_n\ket{n}\bra{n}$, reduces to the eigenstate-resolved expression $T_{\mathrm{mod}}^{(n)} = (k_B\kappa_n)^{-1}$ with $\kappa_n = -\log p_n$. The CP transformation relates the antiparticle state to the particle state via, $\bar{\rho}(\phi) = e^{-i\phi\mathfrak{C}}\rho(0)e^{i\phi\mathfrak{C}}$, where, $\mathfrak{C}$  is acting on off-diagonal phases, so that the off-diagonal elements acquire the phase shifts
$\bar{\rho}_{nm} = \rho_{nm}\,e^{-i\phi(\eta_n - \eta_m)}$ with $\eta_n$ the CP parities. Under time evolution parametrized by proper time $\tau$, the modular temperature evolves as,
\begin{equation}
\frac{dT_{\mathrm{mod}}^{(n)}}{d\tau} =-\frac{(T_{\mathrm{mod}}^{(n)})^2}{k_B}\frac{d\kappa_n}{d\tau}= \frac{(T_{\mathrm{mod}}^{(n)})^2}{k_B p_n}\frac{dp_n}{d\tau}
\label{eq:temp_evolution}
\end{equation}
The eigenvalue evolution receives contributions from both the decay width and the off-diagonal coherences,
\begin{equation}
\frac{dp_n}{d\tau} = \Gamma_n p_n + \sum_{m \neq n}\mathrm{Re}\bigl[(\rho^{\mathrm{off}})_{nm}(\partial_\tau\rho^{\mathrm{off}})_{mn}\bigr]
\label{eq:eigenvalue_evolution}
\end{equation}
The central result connecting modular thermodynamics to CP violation is the following.

\begin{theorem}[Modular Thermal Equilibrium]
CP conservation $(\phi = 0)$ is equivalent to modular thermal equilibrium between particle and antiparticle sectors.
\end{theorem}
\begin{proof}
For the CP-conjugate pair, the off-diagonal contributions to the eigenvalue evolution differ by, 
\begin{equation}
\frac{dp_n^{(\bar{\rho})}}{d\tau} - \frac{dp_n^{(\rho)}}{d\tau}\propto \sin(\phi\,\Delta\eta_{nm})
\label{eq:cp_difference}
\end{equation}
where $\Delta\eta_{nm} = \eta_n - \eta_m$. Define the thermodynamic distance between CP-conjugate modular flows as,
\begin{equation}
\mathbf{D}_{\mathrm{therm}} = \int_0^\infty d\tau\,w(\tau)\sum_n \left|\frac{dT_{\mathrm{mod}}^{(n)}[\rho]}{d\tau}- \frac{dT_{\mathrm{mod}}^{(n)}[\bar{\rho}]}{d\tau}\right|
\label{eq:therm_distance}
\end{equation}
where $w(\tau)$ is a positive weight function. Substituting Eq.~\eqref{eq:cp_difference} into Eq.~\eqref{eq:temp_evolution} and using Eq.~\eqref{eq:therm_distance}, 
\begin{equation}
\mathbf{D}_{\mathrm{therm}} \propto \sin(\phi)\int_0^\infty d\tau\,w(\tau)\sum_{n \neq m}\bigl|(\rho^{\mathrm{off}})_{nm}(\partial_\tau\rho^{\mathrm{off}})_{mn}\bigr|
\label{eq:therm_distance_explicit}
\end{equation}
At $\phi = 0$ the right-hand side vanishes identically, so $\mathbf{D}_{\mathrm{therm}} = 0$, which implies,
\begin{equation}
\frac{dT_{\mathrm{mod}}^{(n)}[\rho]}{d\tau}\bigg|_{\phi=0}= \frac{dT_{\mathrm{mod}}^{(n)}[\bar{\rho}]}{d\tau}\bigg|_{\phi=0}
\end{equation}
Conversely, if $\mathbf{D}_{\mathrm{therm}} > 0$ and the off-diagonal coherences are non-vanishing, then Eq.~\eqref{eq:therm_distance_explicit} requires $\sin(\phi) \neq 0$, hence $\phi \neq 0$.
\end{proof}
The theorem establishes the sharp equivalence:
\begin{equation}
\phi = 0 \Longleftrightarrow\frac{dT_{\mathrm{mod}}[\rho]}{d\tau} = \frac{dT_{\mathrm{mod}}[\bar{\rho}]}{d\tau}
\label{eq:thermal_equilibrium_criterion}
\end{equation}
This is not merely an analogy, the modular temperature is the precise information-theoretic object whose gradient encodes the departure from CP symmetry, and $\mathbf{D}_{\mathrm{therm}} \propto \sin(\phi_{\mathrm{CP}})$ provides a measure of CP violation through modular dynamics alone. The connection to the diffeomorphism invariance established in the previous section is immediate, since $T_{\mathrm{mod}}[\rho]$ is constructed entirely from the
modular operator $K = -\log\rho$, and since both $\mathbf{D}_{\mathrm{asym}}$ and $\mathbf{D}_{\mathrm{sym}}$ are invariant under unitary conjugation of the density matrices, the thermodynamic distance $\mathbf{D}_{\mathrm{therm}}$
inherits the same diffeomorphism invariance as $\mathfrak{A}_{\mathrm{info}}$. CP violation is therefore manifested as a modular thermal disequilibrium that is intrinsic to the quantum state and independent of any choice of spacetime
foliation, coordinate system, or reference frame.

\subsection{Summary}

The construction developed above constitutes a apparent reformulation of CP violation in the language of quantum information geometry, free from the presuppositions of local spacetime structure that render the standard phase-based formulation ill-defined in curved or dynamical backgrounds. The central object $\mathfrak{A}_{\mathrm{info}}[g]$ is a genuine scalar on the space of quantum states, diffeomorphism invariant by the BFV local covariance axioms and the unitary invariance of relative entropy, foliation-independent by the equivalence of any two Cauchy surfaces under internal spacetime diffeomorphisms, and independent of any choice of time coordinate, reference frame, or Fock space representation. This invariance stands in sharp contrast to the standard weak phase, whose operational definition requires stable asymptotic hadronic states, coherent oscillations over macroscopic proper times, and an unambiguous foliation of spacetime, all of which fail in generic curved backgrounds or at GUT-scale energies where no approximately localized hadronic subalgebra exists within the local QFT operator algebra. The algebraic signature of CP violation is the modular commutator $[K, \bar{K}] \neq 0$. The operators $K$ and $\bar{K}$ share an eigenbasis if and only if they commute, which occurs if and only if $\phi= 0$, so the Frobenius norm $\|[K,\bar{K}]\|_F$ provides a basis-independent, unitarily invariant scalar measuring the degree of violation. The antisymmetric relative entropy $\mathbf{D}_{\mathrm{asym}} \neq 0$ if and only if $[K,\bar{K}] \neq 0$, establishing a direct operational bridge between modular non-commutativity and geometric phase. Through the modular temperature $T_{\mathrm{mod}}[\rho] = (\partial S[\rho]/\partial\langle H\rangle)^{-1}$, CP conservation is equivalent to modular thermal equilibrium $dT_{\mathrm{mod}}[\rho]/d\tau = dT_{\mathrm{mod}}[\bar{\rho}]/d\tau$, with the thermodynamic distance $\mathbf{D}_{\mathrm{therm}} \propto\sin(\phi)$ vanishing precisely at the CP-conserving point. This thermodynamic structure leads to a covariant reformulation of Sakharov's conditions in which each of the three standard requirements, formulated in Minkowski spacetime with structures absent in curved backgrounds, admits a precise modular translation, 
\begin{itemize}
\item Baryon number violation corresponds to quantum information flow between subsystems, captured by the non-unitarity of the effective quantum channel $\mathrm{Tr}(\mathcal{E}^\dagger\mathcal{E}) < 1$.
\item C and CP violation correspond to modular flow asymmetry $K \neq \bar{K}$, with the degree of violation measured by the invariant $\|K - \bar{K}\| > 0$, or equivalently by $\mathfrak{A}_{\mathrm{info}}[g] \neq 0$.
\item Departure from thermal equilibrium corresponds to a non-zero modulartemperature gradient $\nabla_\mu T_{\mathrm{mod}} \neq 0$, which in generic curved spacetimes measures the failure of the modular parameter $s \in \mathbb{R}$
to coincide with any geometric time, in the sense of the Connes--Rovelli thermal time hypothesis.
\end{itemize}
Preliminary evidence for a two-component system, with additional supporting derivations in the acknowledgment-linked document, suggests the viability of this approach, though rigorous phenomenological studies are required.

\section{Acknowledgment}
The article is dense in its own right, the readers are encouraged to consult a supplementary note \href{https://www.dropbox.com/scl/fi/p982w8i688w5djrd5744j/Supplementary_DIFCPV.pdf?rlkey=7z0lz14iqpuqp0jg3td6kgjio&st=geq4netg&dl=0}{\color{blue}LINKED} here, a feed back is also expected. The author expresses sincere \& thankful  gratitude to \textbf{Professor Artur Kalinowski} for providing computing resources and extra time outside experimental responsibilities.

\bibliographystyle{elsarticle-num} 
\bibliography{reference}

\end{document}